\documentclass[aps, pra, twocolumn, superscriptaddress, 10pt]{revtex4-2}
\usepackage{preamble}

\begin{document}
\title{On quantum channels that destroy negative conditional entropy}

\author{P V Srinidhi} 
\email{srinidhi.pv@research.iiit.ac.in}

\author{Indranil Chakrabarty}
\email{indranil.chakrabarty@iiit.ac.in}

\author{Samyadeb Bhattacharya}
\email{sbh.phys@gmail.com}

\affiliation{Centre for Quantum Science and Technology, International Institute of Information Technology Hyderabad, Gachibowli, Hyderabad-500032, Telangana, India.\\Center for Security, Theory and Algorithmic Research, International Institute of Information Technology Hyderabad, Gachibowli, Hyderabad-500032, Telangana, India.}

\author{Nirman Ganguly}
\email{nirmanganguly@hyderabad.bits-pilani.ac.in}

\affiliation{Department of Mathematics, Birla Institute of Technology and Science Pilani, Hyderabad Campus, Telangana-500078, India.}

\begin{abstract}
Counter-intuitive to classical notions, quantum conditional entropy can be negative, playing a pivotal role in information-processing tasks. This article delves deeply into quantum channels, emphasizing negative conditional entropy breaking (NCEB) channels and introducing negative conditional entropy annihilating (NCEA) channels. We characterize these channels from topological and information-theoretic perspectives, examining their properties when combined serially. We also examine NCEB channels in parallel. Our exploration extends to complimentary channels associated with NCEB, leading to the introduction of information-leaking channels. Utilizing the parameters of the standard depolarizing channel, we provide tangible examples and further characterization. We demonstrate the relationship of NCEB and NCEA with newly introduced channels like coherent information (CIB) breaking and mutual information (MIB) breaking, along with standard channels like zero capacity channels. Preservation of quantum resources is an integral constituent of quantum information theory. Recognizing this, we lay prescriptions to detect channels that do not break the negativity of conditional entropy, ensuring the conservation of this quantum resource.
\end{abstract}

\maketitle

\section{Introduction} \label{sec:intro}
Quantum resource theory \cite{chitambar2019} is a standard framework to study quantum information theory, where we categorize quantum states based on their usefulness in information processing tasks. Such useful states are termed as resources, and the rest as free states. For example, quantum entanglement \cite{einstein1935, peres1997, bohr1935, bell1964, horodecki2009}, which refers to states that cannot be expressed as a convex combination of product states, is one of the most significant resources in quantum information theory. In the resource theory of entanglement, separable states act as free states. Entanglement is a key ingredient in tasks like teleportation \cite{bennett1993,horodecki1996,chakrabarty2010, sohail2023}, dense coding \cite{bennett1992, srivastava2019, das2015, patro2017, roy2018}, remote state preparation \cite{bennett2005, pati2000}, key generation \cite{bennett2014, zukowski1993, bennett1984, shor2000, ekert1991}, secret sharing \cite{hilery1999, cleve1999, bandyopadhyay2000, adhikari2010probabilistic, ray2016, singh2023,abrol2024secret}, routing quantum information \cite{sazim2015} and setting up quantum network between various quantum processors \cite{wehner2018}.

However, the mere presence of entanglement is not the sole requirement for all quantum tasks. Tasks like state merging \cite{horodecki2006}, dense coding \cite{bennett1992}, quantum memory \cite{berta2010}, and one-way distillation processes \cite{devetak2005} require states with negative conditional entropy. In the resource theory of negative conditional entropy, states with non-negative conditional von Neumann entropy (CVENN) are considered free states. The set of CVENN states for $d \otimes d$ dimensional systems is convex and compact, thereby facilitating the detection of states with negative conditional entropy \cite{vempati2021}. Additionally, states maintaining non-negative conditional entropy even under global unitary action, termed absolute non-negative conditional von Neumann entropy states (ACVENN), have been characterized \cite{patro2017}.

In practical settings, environmental interactions are ubiquitous in information processing tasks, and it is important to understand how a quantum state with some potential as a resource evolves under such conditions. These interactions are modeled as quantum channels, which are trace-preserving completely positive linear map $\mathcal{N}: \rho \rightarrow \sum_i N_i \rho N_i^\dagger$ with $\sum_i N_i^\dagger N_i=I $, and $N_i s$ are the Kraus operators associated with any such map. The noise in such an environment can be significant enough to cause the state to lose its value as a resource. In the resource theory of entanglement, quantum channels that turn any state into a separable state are called Entanglement Breaking (EB) channels. A map $\mathcal{N}$ is entanglement breaking if $(id \otimes \mathcal{N})(\rho)$ is always separable, where $id$ is the identity channel. It was shown that an entanglement breaking channel can be written in the Holevo form \cite{horodecki2003},
\begin{equation}
\mathcal{N}(\rho)= \sum_{k=1}^m R_k Tr(F_k),    
\end{equation}
where $R_k$s are density matrices, and $F_k$s are positive operator value measurement operators (POVMs) and $\sum_k F_k=I$. The choice of $R_k$s and $F_k$s are not unique, and different operators give rise to various types of entanglement-breaking channels. In a similar spirit, we also have Entanglement Annihilating (EA) channels. These channels destroy the entanglement within a subsystem of a composite system \cite{morav_kov2010}.

As noted before, the negativity of quantum conditional entropy can also be considered a resource akin to entanglement. While  EB and EA channels are known to impact quantum states with negative conditional entropy by destroying the entanglement, these channels do not encompass the full scope of the relevant channels. Therefore, quantum channels that affect the negativity of conditional entropy(i.e., convert any state to a state with non-negative conditional entropy) without necessarily disturbing entanglement represent distinct classes of channels and warrant attention. We refer to channels that destroy the negative conditional entropy between a bipartition as Negative Conditional Entropy Breaking (NCEB). In a recent article, information-theoretic resource-breaking channels pertaining to negative conditional entropy and fully entangled fraction have been introduced \cite{muhuri2023}. The authors have characterized such channels and provided methods to detect channels that are not NCEB. On the other hand, Negative Conditional Entropy Annihilating (NCEA) refers to the class of channels that destroy the negative conditional entropy within a system.

In this article, we take the characterization of NCEB further. Besides the usual topological characterization of NCEB channels, we formulate an information-theoretic characterization of such channels. In particular, we study the complementary channel of NCEB channels to discover their information-leaking property. We also exemplify these channels in terms of the range of the parameters of depolarizing channels. We show that the set of NCEB channels is equivalent to channels with zero coherent information. In addition, we show that the set of  NCEB channels acts as a superset to EB channels, zero capacity channels, and the newly introduced Mutual Information Breaking (MIB) channels. Furthermore, we introduce the notion of NCEA channels, which differ from NCEB and provide for their characterization. We exemplify NCEA channels through global depolarizing channels and produce a sufficiency condition for them to become negative conditional entropy annihilating. We also connect NCEA channels to conditional Von-Neumann entropy non-decreasing (NCVE) channels. In a related perspective, when implementing a quantum information processing protocol, one would like to know which channels to select that do not result in a resource loss. It is here that the detection of non-NCEB and non-NCEA channels becomes imperative. We lay down prescriptions for the detection of non-NCEB channels.

The remainder of this article is structured as follows: In Section \ref{sec:preliminaries}, we present the necessary preliminary concepts. Section \ref{sec:NCEBA} provides definitions of NCEB and NCEA channels, examples, and properties. In Section \ref{sec:relations}, we explore the relationships between NCEB, NCEA, and other known channels, including the newly introduced CIB and MIB channels. Section \ref{sec:convex_compact} focuses on the topological characterization of NCEB and NCEA channels, enabling the detection of non-NCEB channels. Finally, we conclude with a summary of our findings in Section \ref{sec:conclusion}.

\section{Preliminaries} \label{sec:preliminaries}
A quantum system $A$ is associated with a Hilbert space $\mathcal{H}_A$. Similarly, a bipartite system $AB$ 
is described by a Hilbert space $\mathcal{H}_{AB} = \mathcal{H}_A \otimes \mathcal{H}_B$. A quantum system isomorphic to $A$ is denoted by $\tilde{A}$, with $\mathcal{H}_{\tilde{A}}$ as its Hilbert space. $B(\mathcal{H})$ represents the set of bounded linear operators on $\mathcal{H}$, and $B_{+}(\mathcal{H})$ is the subset of positive semi-definite linear operators. The set $D(\mathcal{H})$ denotes the collection of unit-trace, positive semi-definite linear operators, commonly known as density operators. Thus, the state of a quantum system $AB$ is represented by a density operator $\rho_{AB} \in D(\mathcal{H}_{AB})$. Lastly, the identity operator is given by $I$, and the identity channel is given by $id$.

\subsection{Information measures on quantum states} \label{sec:preliminaries_info_measures}
The Von Neumann entropy of a quantum system with density operator $\rho$ is defined as
\begin{equation}
    S(\rho) = - \text{Tr}\{\rho \log \rho \},
\end{equation} 
where the base of the logarithm is $2$. Quantum conditional entropy measures the uncertainty associated with one part of a bipartite system (e.g., system $A$) when given complete knowledge of the other part (system $B$). Given the density operator $\rho_{AB}$ of a bipartite system, the conditional entropy of the state is expressed as
\begin{equation} 
    S(A|B)_\rho = S(AB)_\rho - S(B)_\rho
\end{equation} 
with $S(B)_\rho$ as the entropy of the marginal operator $\rho_B = \text{Tr}_A\{\rho_{AB}\}$. The expression is derived from the chain rule, $S(AB)_\rho = S(A|B)_\rho + S(B)_\rho$. Conditional entropy is uniformly continuous with respect to the trace norm \cite{alicki2004, winter2016}. Additionally, the negativity of conditional entropy is a signature of quantum systems, as classically, this is impossible. Negative values of conditional entropy indicate the presence of entanglement in the system; hence, conditional entropy itself is an important feature in several information processing protocols. However, entanglement does not necessitate quantum systems to possess negative conditional entropy.

Likewise, quantum mutual information measures the total classical and quantum correlations present in the state. For a given bipartite system described by its density operator $\rho_{AB}$, the mutual information between the subsystems $A, B$ is defined as
\begin{equation}
\begin{split}
    I(A;B) & = S(A)_\rho + S(B)_\rho - S(AB)_\rho\\
    & = S(A)_\rho - S(A|B)_\rho\\
    & = S(B)_\rho - S(B|A)_\rho.
\end{split}
\end{equation} Quantum mutual information is always non-negative and vanishes if and only if the subsystems are uncorrelated. It is a symmetric measure, $I(A; B) = I(B; A)$. Quantum mutual information has been extended to multiparty scenarios to provide measures for quantum correlation \cite{chakrabarty2011,modi2012}. Additionally, this measure underpins specific capacity measures for quantum channels, such as Holevo capacity and entanglement-assisted classical capacity \cite{wilde_qit2016}.
 
\subsection{Coherent Information of Quantum Channels} \label{sec:preliminaries_coh_info}
Quantum channels are mathematical models that describe the evolution of quantum systems. They can also be viewed as a model of the communication medium through which one party, A(Alice), transmits information in the form of a quantum system to the other party, B(Bob). Analogous to classical channels, we can quantify the information-carrying capacity of a quantum channel under different settings via capacity measures. Surveys and discussions of such capacity measures have been conducted in \cite{smith2010quantum, watrous2018, gyongyosi2018, holevo2020}. One such measure, namely the coherent information of a quantum channel, captures the achievable rate of quantum information transmission through a quantum channel in a reliable manner. Consider a finite-dimensional quantum system $AA'$, with $\mathcal{N}:B(\mathcal{H}_{A'}) \rightarrow B(\mathcal{H}_B)$ ( $\mathcal{N}_{A' \rightarrow B}$) representing a quantum channel into the finite-dimensional space $\mathcal{H}_B$. Then, the coherent information of the channel $\mathcal{N}_{A' \rightarrow B}$ is 
\begin{equation}
\label{eqn:coh_inf_chnl}
Q(\mathcal{N}_{A' \rightarrow B}) = \max_{\phi_{AA'}} I(A \rangle B)_\sigma,    
\end{equation}
where $\sigma_{AB} = \mathbb{I}_A \otimes \mathcal{N}_{A'\rightarrow B}(\phi_{AA'})$ is the channel output for the pure state input $\phi_{AA'}$.  $I(A \rangle B)_\sigma = S(B)_\sigma - S(AB)_\sigma$ captures the coherent information of the output state. The coherent information of the quantum channel can also be expressed as,
\begin{equation}
\begin{split}
  Q(\mathcal{N}_{A' \rightarrow B}) & = \max_{\rho_{A'}} I_C(\rho_{A'}, \mathcal{N}_{A'\rightarrow B})\\
  & = \max_{\phi_{AA'}}[S(B)_\psi - S(E)_\psi],
\end{split}
\end{equation}
where 
\begin{equation}
\begin{split}
    I_C(\rho_{A'}, \mathcal{N}_{A'\rightarrow B}) &= \\
    &S(\mathcal{N}_{A'\rightarrow B}(\rho_{A'}) - S(\mathcal{N}^C_{A'\rightarrow B}(\rho_{A'})))
\end{split}
\end{equation} with $\mathcal{N}^C_{A'\rightarrow B}$ as the complementary channel to $\mathcal{N}_{A'\rightarrow B}$ and $|\psi_{ABE}\rangle = U^{\mathcal{N}}_{A'\rightarrow BE} |\phi_{AA'}\rangle$ using $U^{\mathcal{N}}_{A'\rightarrow BE}$, the isometric extension of the channel \cite{wilde_qit2016}. Finally, the quantum capacity of a quantum channel \cite{devtak2005Channel, barnum1998, barnum2000, shorSLMath2002, llyod1997} captures the information-carrying capacity of the channel in an asymptotic setting and is formally defined as
\begin{equation}
     \mathscr{Q}(\mathcal{N}) = \lim_{n \rightarrow \infty} \frac{1}{n} Q(\mathcal{N}^{\otimes n}).
\end{equation}

\section{Negative Conditional Entropy Annihilating and Breaking Channels}\label{sec:NCEBA}
We introduce Negative Conditional Entropy Annihilating (NCEA) channels and discuss previously introduced Negative Conditional Entropy Breaking (NCEB) channels \cite{muhuri2023}. Like entanglement breaking channels, NCEB channels affect the conditional entropy across the bipartition of a composite system. On the other hand, NCEA channels affect the conditional entropy of the system they act on and may not affect the same property across partitions. Section \ref{sec:NCEBA_Def} covers the definitions of both NCEB and NCEA channels. The properties of such channels are discussed in section \ref{sec:NCEBA_Props}. Here, we investigate the behavior of these channels in series and NCEB channels in parallel combinations. Furthermore, we consider the action of complementary channels to NCEB to reveal their effects on the environment or adversary coupled with the original system. Finally, we provide examples of NCEB and NCEA channels in section \ref{sec:NCEBA_Ex}.

In later sections, we formulate relations of NCEB and NCEA with other known classes of quantum channels. Furthermore, we discuss previous results  \cite{muhuri2023} for a complete and pedagogic discussion. Our treatment of NCEB channels presents perspectives that complement the previous work. 
\setlength{\fboxsep}{3pt} 
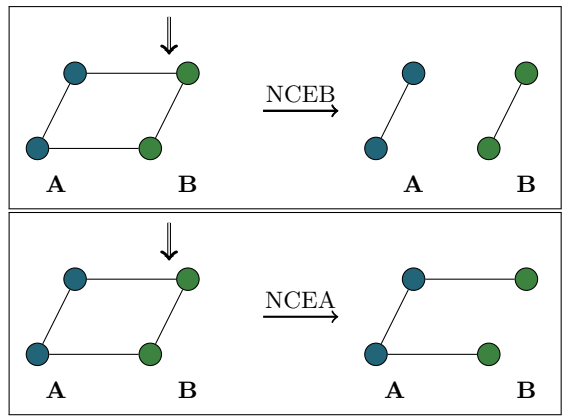
\begin{figure}
\begin{subfigure}[b]{0.5\textwidth}
   \fbox{
        \begin{tikzpicture}[scale=0.5]
            \node[draw, circle, fill={rgb:red,41; green,120; blue,142}, inner sep=3pt] (A) at (0,0) {};
            \node[draw, circle, fill={rgb:red,94; green,205; blue,101}, inner sep=3pt] (B) at (3,0) {};
            \node[draw, circle, fill={rgb:red,41; green,120; blue,142}, inner sep=3pt] (C) at (-1,-2) {};
            \node[draw, circle, fill={rgb:red,94; green,205; blue,101}, inner sep=3pt] (D) at (2,-2) {};
            \draw (A) -- (B);
            \draw (A) -- (C);
            \draw (B) -- (D);
            \draw (C) -- (D);
            \draw[->] (2.5, 1.5) -- (2.5, 0.5);
            \draw[double, ->] (2.5, 1.5) -- (2.5, 0.5);
            \node[below] at (-0.5, -2.5) {\textbf{A}};
            \node[below] at (3, -2.5) {\textbf{B}};
            \draw[->, thick] (5,-1) -- (7,-1) node[midway, above] {NCEB};
            \node[draw, circle, fill={rgb:red,41; green,120; blue,142}, inner sep=3pt] (F) at (9,0) {};
            \node[draw, circle, fill={rgb:red,94; green,205; blue,101}, inner sep=3pt] (G) at (12,0) {};
            \node[draw, circle, fill={rgb:red,41; green,120; blue,142}, inner sep=3pt] (H) at (8,-2) {};
            \node[draw, circle, fill={rgb:red,94; green,205; blue,101}, inner sep=3pt] (I) at (11,-2) {};
            \draw (F) -- (H);
            \draw (G) -- (I);
            \node[below] at (9, -2.5) {\textbf{A}};
            \node[below] at (12, -2.5) {\textbf{B}};
            \draw[->, thick] (5,-1) -- (7,-1);
        \end{tikzpicture}
    }
    \phantomcaption{}
    \label{subfig:NCEB_Channel}
\end{subfigure}
\begin{subfigure}[b]{0.5\textwidth}
    \fbox{
        \begin{tikzpicture}[scale=0.5]
            \node[draw, circle, fill={rgb:red,41; green,120; blue,142}, inner sep=3pt] (A) at (0,0) {};
            \node[draw, circle, fill={rgb:red,94; green,205; blue,101}, inner sep=3pt] (B) at (3,0) {};
            \node[draw, circle, fill={rgb:red,41; green,120; blue,142}, inner sep=3pt] (C) at (-1,-2) {};
            \node[draw, circle, fill={rgb:red,94; green,205; blue,101}, inner sep=3pt] (D) at (2,-2) {};
            \draw (A) -- (B);
            \draw (A) -- (C);
            \draw (B) -- (D);
            \draw (C) -- (D);
            \draw[->] (2.5, 1.5) -- (2.5, 0.5);
            \draw[double, ->] (2.5, 1.5) -- (2.5, 0.5);
            \node[below] at (-0.5, -2.5) {\textbf{A}};
            \node[below] at (3, -2.5) {\textbf{B}};
            \draw[->, thick] (5,-1) -- (7,-1) node[midway, above] {NCEA};
            \node[draw, circle, fill={rgb:red,41; green,120; blue,142}, inner sep=3pt] (F) at (9,0) {};
            \node[draw, circle, fill={rgb:red,94; green,205; blue,101}, inner sep=3pt] (G) at (12,0) {};
            \node[draw, circle, fill={rgb:red,41; green,120; blue,142}, inner sep=3pt]  (H) at (8,-2) {};
            \node[draw, circle, fill={rgb:red,94; green,205; blue,101}, inner sep=3pt] (I) at (11,-2) {};
            \draw (F) -- (H);
            \draw (F) -- (G);
            \draw (H) -- (I);
            \node[below] at (8.5, -2.5) {\textbf{A}};
            \node[below] at (12, -2.5) {\textbf{B}};
        \end{tikzpicture}
    }
    \phantomcaption{}
    \label{subfig:NCEA_Channel}
\end{subfigure}
\caption{The figure illustrates the effects of two distinct types of quantum channels on a bipartite quantum system $AB$. (a) The Negative Conditional Entropy Breaking (NCEB) channel affects the correlations between subsystems A and B, breaking the entanglement across the bipartition. (b) The Negative Conditional Entropy Annihilating (NCEA) channel, on the other hand, impacts the correlations within the specific subsystem it acts on.}
\end{figure}

\subsection{Definitions} \label{sec:NCEBA_Def}
Given a $d \otimes d$ bipartite system $AB$, let $\mathcal{S}_{CVENN}(\mathcal{H}_{AB})$ denote the set of quantum states with non-negative conditional entropy between $A$ and $B$. Similarly, let $\mathcal{S}_{CVENN}(\mathcal{H}_{B})$ be the set of quantum states having non-negative conditional entropy within the subsystem $B$. This is expressed as
\begin{equation}
\begin{split}
    \mathcal{S}_{CVENN}(\mathcal{H}_{AB}) = \{ \rho  \in D(\mathcal{H}_{AB}) | S(A|B)_\rho \ge 0 \}, \\ \mathcal{S}_{CVENN}(\mathcal{H}_B) = \{ \rho \in D(\mathcal{H}_B) | S(B_1 | B_2)_\rho \ge 0\}
\end{split}
\end{equation} where $B_1, B_2$ are a fixed partitions of the $B$ subsystem.

As depicted in the figure \ref{subfig:NCEB_Channel}, an NCEB channel acts on one part of a bipartite system $AB$ (e.g., system $B$) and destroys the negative conditional entropy present between $A$ and $B$ subsystems. In other words, the set of negative conditional entropy-breaking channels acting on $d \otimes d$-dimensional system is described as,
\begin{equation}
\begin{split}
     NCEB^{(d)} & = \{\mathcal{N}_{B\rightarrow \tilde{B}}| \forall \rho_{AB} \in D(\mathcal{H}_{AB}), S(A|\tilde{B})_\sigma \ge 0, \\&  \text{where }  \sigma_{A\tilde{B}} = id_A \otimes \mathcal{N}_{B\rightarrow \tilde{B}}(\rho_{AB})\}.
\end{split}
\end{equation} Equivalently, a channel $\mathcal{N}_{B\rightarrow\tilde{B}}$ is in  $NCEB^{(d)}$ if
\begin{equation}
(id_A \otimes  \mathcal{N}_{B\rightarrow\tilde{B}})(D(\mathcal{H}_{AB})) \subset \mathcal{S}_{CVENN}(\mathcal{H}_{A\tilde{B}}).   
\end{equation}

Likewise, NCEA channels act on subsystem $B$ and annihilate the negative conditional entropy within the $B$ (refer figure \ref{subfig:NCEA_Channel}). For a $d$-dimensional system, the set of negative conditional annihilating channels is expressed as
\begin{equation}
\begin{split}
    NCEA^{(d)} & = \{\mathcal{N}_{B\rightarrow \tilde{B}}| \forall \rho_{B} \in D(\mathcal{H}_{B}), S(\tilde{B}_1|\tilde{B}_2)_\epsilon \ge 0, \\
    & \text{ where }  \epsilon_{\tilde{B}} =  \mathcal{N}_{B\rightarrow \tilde{B}}(\rho_{B}) \}.
\end{split}
\end{equation} From the definition above, it follows that a channel $\mathcal{N}_{B\rightarrow\tilde{B}}$ is $NCEA$ if 
\begin{equation}
    \mathcal{N}_{B\rightarrow\tilde{B}}(D(\mathcal{H}_B)) \subset \mathcal{S}_{CVENN}(\mathcal{H}_{\tilde{B}}).   
\end{equation}

\subsection{Properties} \label{sec:NCEBA_Props}
\noindent \textbf{I. NCEB/NCEA channels in series:} Given two quantum channels $\mathcal{N}_1, \mathcal{N}_2$, let the serial combination of these channels be denoted with $\mathcal{N}_1 \circ \mathcal{N}_2$ and expressed as $\mathcal{N}_1 \circ \mathcal{N}_2(\rho) = \mathcal{N}_1(\mathcal{N}_2(\rho)))$ for input state $\rho$.  We show that for a serial combination of two channels $\mathcal{N}_1$ and $\mathcal{N}_2$ both taken from  $NCEB^{(d)}$ or $NCEA^{(d)}$, the resultant channel will always belong to the same set.\\ 

\begin{theorem}
Let  $\mathcal{N}_1, \mathcal{N}_2 \in NCEB^{(d)}$ , then  $\mathcal{N}_1 \circ \mathcal{N}_2 \in NCEB^{(d)}$.   
\end{theorem}

\begin{proof}
    Let $\mathcal{N}_1$ and $\mathcal{N}_2$ be quantum channels belonging to $NCEB^{(d)}$. The action of the serial combination of these channels on an input state $\rho$ is given by $\mathcal{N}_1(\mathcal{N}_2(\rho))$. Thus, 
    \begin{equation}
    \begin{split}
        (id_A \otimes \mathcal{N}_1)(D(\mathcal{H}_{AB})) \subset \mathcal{S}_{CVENN}(\mathcal{H}_{AB}), \\
        (id_A \otimes \mathcal{N}_2)(D(\mathcal{H}_{AB})) \subset \mathcal{S}_{CVENN}(\mathcal{H}_{AB}).
    \end{split}
    \end{equation}
    It follows that, 
    \begin{equation}
    \begin{split}
        (id_A \otimes \mathcal{N}_1 \circ \mathcal{N}_2)(D(\mathcal{H}_{AB})) & = \\
        (id_A \otimes \mathcal{N}_1) \circ (id_A \otimes \mathcal{N}_2)(D(\mathcal{H}_{AB}))
    \end{split}
    \end{equation} and since the range of $(id_A \otimes \mathcal{N}_2)$ is within $\mathcal{S}_{CVENN}(\mathcal{H}_{AB})$, the range of $(id_A \otimes \mathcal{N}_2 \circ \mathcal{N}_1)$ must also be within $\mathcal{S}_{CVENN}(\mathcal{H}_{AB})$ (follows from the definition of NCEB). Thus, $\mathcal{N}_2 \circ \mathcal{N}_1 \in NCEB^{(d)}$
\end{proof}

\begin{theorem}
Let  $\mathcal{N}_1, \mathcal{N}_2 \in NCEA^{(d)}$ , then  $\mathcal{N}_1 \circ \mathcal{N}_2 \in NCEA^{(d)}$.   
\end{theorem}
\begin{proof}
    Let $\mathcal{N}_1$ and $\mathcal{N}_2$ be quantum channels belonging to $NCEA^{(d)}$ and consider the serial combination $\mathcal{N}_1 \circ \mathcal{N}_2$  The action of the this combination on input state $\rho$ is expressed as $\mathcal{N}_1(\mathcal{N}_2(\rho))$. From this, we have, $\mathcal{N}_1(D(\mathcal{H}_{B}) \subset \mathcal{S}_{CVENN}(\mathcal{H}_{B})$ and $\mathcal{N}_2(D(\mathcal{H}_{B}) \subset \mathcal{S}_{CVENN}(\mathcal{H}_{B})$ which follows from the definition of NCEA channels. It follows that, $(\mathcal{N}_1 \circ \mathcal{N}_2)(D(\mathcal{H}_{B})) \subset \mathcal{S}_{CVENN}(\mathcal{H}_{B})$ and thus, $\mathcal{N}_1 \circ \mathcal{N}_2 \in NCEA^{(d)}$.
\end{proof}

\noindent \textbf{II. NCEB channels in parallel:} We consider the parallel combination of NCEB channels and show that the resultant channel is not necessarily NCEB. The argument is built upon the connections between NCEB channels, zero quantum capacity channels (refer section \ref{sec:relations_ZC}), and the coherent information breaking channels (section \ref{sec:relations_CIB}). This equivalence allows us to extend properties and results related to zero coherent information channels to the set of NCEB channels.

This claim can be supported using superactivation of quantum capacity, particularly for channels with zero quantum capacity. Some well-known examples of zero capacity channels include symmetric quantum channels \cite{bennett1997} and entanglement binding channels \cite{horodecki1996, horodecki1997}(also known as Horodecki channels). Though symmetric channels display a correlation between input and output, they are useless in sending quantum information, as a non-zero transfer would violate the no-cloning theorem. On the other hand, Horodecki channels can only produce weakly entangled states satisfying the PPT criterion, resulting in their inability to transmit information. 

For channels $\mathcal{N}_1 \in N_H, \mathcal{N}_2 \in A_S$, where $N_H$ and $A_S$ are the set of Horodecki channels and symmetric channels, we have $\mathscr{Q}(\mathcal{N}_1)= 0 =\mathscr{Q}(\mathcal{N}_2)$. It then follows that $Q(\mathcal{N}_1) = 0 = Q(\mathcal{N}_2)$ as $\mathscr{Q}(\mathcal{N}) \ge Q(\mathcal{N}) \ge 0$ for any quantum channel $\mathcal{N}$. It is known that  $Q(\mathcal{N}_1 \otimes \mathcal{N}_2) > 0$ \cite{smith2008}, implying that $\mathcal{N}_1 \otimes \mathcal{N}_2$ is not an NCEB channel. Consequently, from both lines of argument, it is evident that the parallel combination of NCEB channels is not guaranteed to be NCEB

Additionally, we can consider the example of a qubit depolarizing channel $\mathcal{N}$ with fidelity parameter $f$. It is known that the channel cannot transmit information for $f < 0.81$ \cite{divincenzo_1998}. In other words, the coherent information $Q(\mathcal{N}) = 0$ for $f < 0.81$. This implies that the channel is NCEB for $0 \le f < 0.81$. Let  $\mathcal{N}^{\otimes 5}$ represent a parallel combination of five qubit depolarizing channels. Based on the above arguments, it is expected that $\mathcal{N}^{\otimes 5}$ will be an NCEB channel for $f < 0.81$. However, there exists a certain additive quantum coding scheme, which can lower the transmission threshold value for the $\mathcal{N}^{\otimes 5}$. In this case, the data encoding scheme resulted in information transmission for $f > 0.809$. This example suggests the existence of coding schemes for which a parallel combination of NCEB channels might produce positive coherent information, even if the individual channels cannot do so.

\noindent \textbf{III. Complementary of NCEB channel:} We consider an isometric extension of a negative conditional entropy breaking channel $\mathcal{N}_{B \rightarrow \tilde{B}}$ (refer figure \ref{fig:isomteric_ext}). Let $U^{\mathcal{N}}_{B \rightarrow \tilde{B}E}$ the isometric extension of $\mathcal{N}_{B \rightarrow \tilde{B}}$ and treat $E$ as environment. The complementary channel $\mathcal{N}^{C}_{B \rightarrow E}$ is a quantum channel from $B$ to $E$ given by
\begin{equation}
\mathcal{N}^{C}_{B \rightarrow E}(\rho)=Tr_{\tilde{B}}\{U^{\mathcal{N}}_{B \rightarrow \tilde{B}E}(\rho)\},    
\end{equation}
for any input quantum state $\rho \in D(\mathcal{H}_B)$. We establish some aspects of the complementary channel below.
\begin{figure}[htbp]
    \fbox{\includegraphics[width=8.5cm]{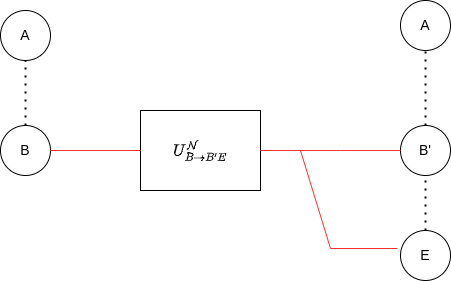}}
    \caption{Isometric Extension of an NCEB channel with $AB'$ as the output system and $E$ the environment. Under the action of an NCEB channel, any negative conditional entropy between $AB$ is transferred to $AE$ at the output end, implying an information leak due to the action of the channel.}
    \label{fig:isomteric_ext}
\end{figure}


 \begin{theorem}
 \label{thm:compl_cond_ent}
 For any channel $\mathcal{N}_{B \rightarrow \tilde{B}} \in NCEB^{(d)}$ acting on an state $\rho_{AB}$, the conditional entropy between systems $A$ and $E$ is always non-positive i.e $S(A|E)_\psi \le 0$ where $|\psi\rangle_{A\tilde{B}E} $ is a purification of the output system.
\end{theorem}
\begin{proof}
    The action of an NCEB channel $\mathcal{N}_{B \rightarrow \tilde{B}}$ on the state  $\rho_{AB}$ can be written as
    \begin{equation}
    (id_A \otimes \mathcal{N}_{B \rightarrow \tilde{B}})(\rho_{AB}) = \sigma_{A\tilde{B}}.  
    \end{equation}
    Thus, we have $S(A|\tilde{B})(\sigma) \ge 0$. Consider a purification $|\psi\rangle_{A\tilde{B}E}$ for $\sigma_{A\tilde{B}}$. It follows that
    \begin{equation}
    \begin{split}
        S(A|\tilde{B})_\sigma & = S(A\tilde{B})_\sigma -S(\tilde{B})_\sigma \\
        & = S(A\tilde{B})_\psi - S(\tilde{B})_\psi \\
        & = S(E)_\psi- S(AE)_\psi \\
        &  = - S(A|E)_\psi.
    \end{split}
    \end{equation}
    If $S(A|\tilde{B})_\sigma \ge 0$, which holds true under the action of an NCEB channel, then $S(A|E)_\psi \le 0$.    
\end{proof} Now, if we consider $E$ as an adversary, Eve, we can prove that Eve can leak out information about the input system $AB$ under the action of these channels. In other words, the mutual information $I(A; E)$ about the system $A$ and the Eve will always be greater than the output mutual information $I(A;\tilde{B})$. 

\begin{theorem}
\label{thm:compl_mut_inf}
 $I(A;\tilde{B}) \le I(A;E)$  
\end{theorem}
\begin{proof} We know that under the action of a channel $\mathcal{N}_{B \rightarrow \tilde{B}} \in NCEB^{(d)}$, the output conditional entropy to be as $S(A|\tilde{B})\sigma \ge 0$. For such channels, the complementary channel $\mathcal{N}^C_{B \rightarrow E}$ produces conditional entropy  $S(A|E)_\psi \le 0$. If we consider the difference
\begin{eqnarray}
&&I(A; \tilde{B})_\psi - I(A; E)_\psi {}\nonumber\\&&
= S(A|E)_\psi -  S(A|\tilde{B})_\psi \le 0 {}\nonumber\\&&
\implies I(A; \tilde{B})_\psi \le I(A; E)_\psi,
\end{eqnarray}
where inequality follows from the fact that $S(A|E)_\psi \le 0$ while $S(A|\tilde{B})_\psi \ge 0$.
\end{proof} The above results reveal a fascinating aspect of NCEB channels. The complementary channel of NCEB  can be interpreted not only as \textit{information leaking} channel but also as \textit{hacking} channel from Eve's point of view. A channel is noiseless if the environment gains any information about the state being transferred. In other words, the channel does not leak information.  When NCEB channels act on one part, say $B$, of a quantum system $AB$, the mutual information between systems $A$ and $\tilde{B}$ is less than the mutual information between $A$ and environment $E$. Thus, the channel allows the environment to have more information about $A$ than $\tilde{B}$, which can equivalently be taken as the information leaking aspect of the channel. Consider a scenario where the environment $E$ is nothing but Eve coupled with the system $AB$. With the action of an NCEB channel, Eve can have more information about $A$ than $\tilde{B}$. In this sense, Eve has hacked the information. This prompts us to term these channels as hacking or information-leaking channels. In the next section, we revisit the complement of these channels when covering examples of NCEB channels.

\subsection{Examples} \label{sec:NCEBA_Ex}
We cover depolarizing channels in the context of NCEB and NCEA channels. Negative conditional entropy is present in quantum systems as small as $2 \otimes 2$ systems. Consequently, we exemplify NCEB channels through qubit depolarizing channels. Similarly, we study global depolarizing and transpose depolarizing quantum channels in the context of NCEA channels and identify sufficiency conditions for which they become NCEA.\\

\noindent \textbf{I: Example of NCEB channel:}\label{sec:NCEBA_Ex_NCEB} Depolarizing channels are a well-known type of quantum channel in quantum information. The action of this channel on a $d$-dimensional system, described by its density operator $\rho$, is expressed as
\begin{equation}
\label{eqn:qubit_dep}
    \mathcal{N}_p^d(\rho) = (1 - p) \rho + p \frac{I}{d} Tr\{\rho\},
\end{equation} where $p$ ($0 \le p \le 1$) is the mixing parameter. We obtain a qubit depolarizing channel when $d= 2$ and denote it as $\mathcal{N}^2_p$.

Negative conditional entropy is a feature of entanglement. It follows that entanglement-breaking channels will break negative conditional entropy and provide direct examples for NCEB channels. $\mathcal{N}^2_p$ becomes entanglement breaking for $\frac{2}{3} \le p \le 1$. This is determined using the sufficiency condition for entanglement breaking channels \cite{horodecki2003} and the positive partial transpose (PPT) test \cite{Horodecki1996Sep, peres1996}. However, entangled quantum states with positive conditional entropy exist. Therefore, there must be a larger range of $p$ for which $\mathcal{N}^2_p$ is an $NCEB$ channel.

To identify the range of $p$ for which the channel becomes NCEB, it suffices to evaluate the action of $\mathcal{N}^2_p$ on one part of $2 \otimes 2$ pure entangled states (equation \ref{eqn:pure_ent}). The supporting argument relies on the equivalence between NCEB channels and Coherent Information Breaking (CIB) channels (refer \ref{sec:relations_CIB}). For the qubit depolarizing channel $\mathcal{N}^2_p$, computing the coherent information of the channel can be simplified to optimization over pure entangled states because of the unitary covariant nature of the channel, i.e., $\mathcal{N}^2_p(U \rho U^\dagger) =  U \mathcal{N}^2_p(\rho) U^\dagger$ \cite{wilde_qit2016}. From a different perspective, the coherent information of a quantum channel can be connected to the hashing point of the channel. The hashing point of a quantum channel is the threshold value of fidelity (noise) parameter below (above) which the information transmission capacity of the channel becomes zero \cite{divincenzo_1998, wilde_qit2016}. Thus, the hashing point reveals the value of $p$ above which the channel becomes NCEB.

Likewise, complementary channels to NCEB display properties covered in theorems \ref{thm:compl_cond_ent} and \ref{thm:compl_mut_inf}. In this context, the channels complementary to qubit depolarizing channels have been studied \cite{Leung_2017}. It is established that such channels possess positive coherent information in all $p < 1$. Thus, for a depolarizing channel in the NCEB regime, their complementary channels will have positive coherent information and result in output state $\sigma$ such that $I(A; E)_\sigma \ge I(A;\tilde{B})_\sigma$\\

\noindent \textit{a. Non-maximally entangled state: }We examine $2 \otimes 2$ pure entangled states of form
\begin{equation}
\label{eqn:pure_ent}
     \ket{\psi} = \cos \alpha \ket{00} + \sin \alpha \ket{11},
\end{equation} 
where $\alpha \in [0, \pi]$. The action of $id \otimes \mathcal{N}^2_p$ on $\ket{\psi}$ produces the following state
\begin{equation}
    \rho_{AB} = \begin{bmatrix}
         \beta_1 & 0 & 0 & \beta_2\\
        0 & \frac{p}{2} \cos^2 \alpha & 0 & 0 \\
        0 & 0 & \frac{p}{2} \sin^2 \alpha & 0 \\
        \beta_2 & 0 & 0 & (1 - \frac{p}{2}) - \beta_1
    \end{bmatrix}.
\end{equation} with $\beta_1 = (1 - \frac{p}{2}) \cos^2 \alpha$ and $\beta_2 = \frac{(1 - p)}{2} \sin 2\alpha$. The conditional entropy $S(A|B)_\rho$ is then given by
\begin{equation}
\begin{split}
    S(A|B)_\rho = & -\lambda_1 \log \lambda_1 -\lambda_2 \log \lambda_2 -\lambda_3 \log \lambda_3 \\ 
    & - \lambda_4 \log \lambda_4 + (\lambda_5)\log(\lambda_5) \\ 
    & + (1 - \lambda_5)\log(1 - \lambda_5),
\end{split}
\end{equation} where 
\begin{equation}
\begin{split}
     \lambda_1 = \frac{p}{2} \cos^2 \alpha, & \lambda_2 = \frac{p}{2} \sin^2\alpha, \\
     \lambda_3 = \frac{1}{2} - \sigma_1 - \frac{p}{4}, & \lambda_4 = \frac{1}{2} + \sigma_1 - \frac{p}{4}, \\
     \text{and } \lambda_5 = cos^2\alpha - \frac{p}{2} \cos 2\alpha, \\
\end{split}
\end{equation} with $\sigma_1 = \frac{1}{8} \sqrt{5 p^2 - 12 p + 8 + 4 p \cos 4 \alpha - 3p^2 \cos 4 \alpha}$.\\

Numerical simulations, illustrated in figure \ref{fig:Qubit_depolarizing_entangled_state}, indicate that maximally entangled states maximize the conditional entropy for a given $p$. This suggests that checking the output on maximally entangled states is sufficient to determine if $\mathcal{N}^2_p$ is NCEB. This observation is supported by Theorem 1 in the related work(\cite{muhuri2023}), which covers a sufficiency condition for group covariant channels to be NCEB. Figure \ref{fig:Qubit_depolarizing_entangled_state} depicts that $\mathcal{N}^2_p$ becomes NCEB for $p > 0.2$. It can be inferred that $\mathcal{N}^2_p$ where $0.2 < p < 0.66 (\frac{2}{3})$ are examples of NCEB channels that are not EB.\\
\begin{figure}[htbp]
    \hspace{-1cm} \includegraphics[width=1.1\columnwidth]{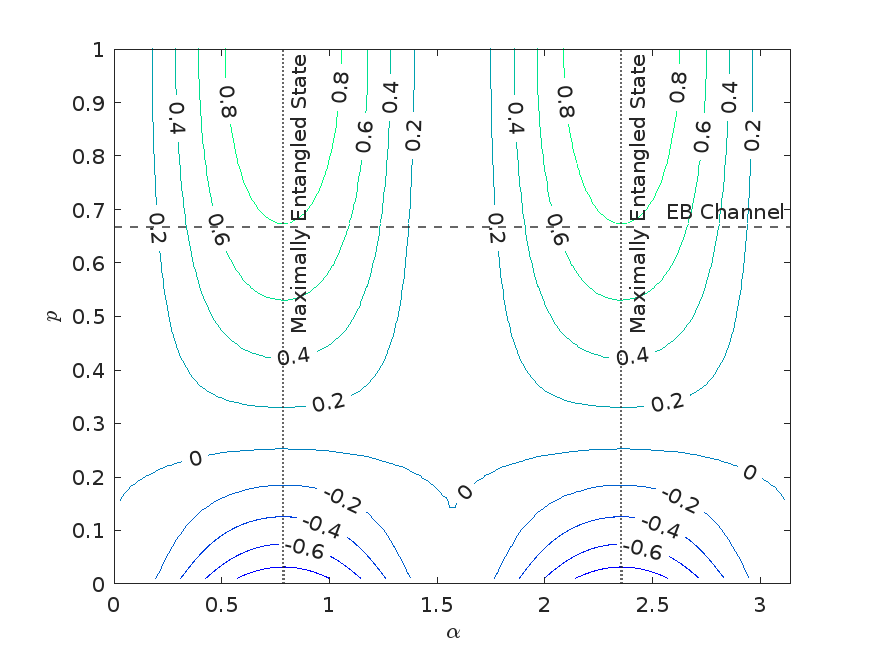}
    \caption{Conditional entropy of 2-qubit pure entangled states after the action of $id \otimes \mathcal{N}^2_p$ as a contour plot. The state parameter $\alpha$ is mapped to the x-axis and the channel parameter $p$ to the y-axis, with curves representing the level sets of the conditional entropy. The horizontal dashed line represents the minimum $p$ value ($p = \frac{2}{3}$) for which the $\mathcal{N}^2_p$ becomes entanglement breaking. More generally, the plot illustrates that maximally entangled states ($\alpha = \frac{\pi}{4} (0.785)$ or $\alpha = \frac{3 \pi}{4} (2.356)$) demand higher $p (> 0.2)$ values to break negative conditional entropy compared to other pure entangled states.}
  \label{fig:Qubit_depolarizing_entangled_state}
\end{figure}

\noindent \textit{b. Bell-diagonal state: }We are interested in the action of qubit depolarizing channels on general $2 \otimes 2$ quantum states. Consider the set of all \textit{Bell diagonal} states, which is given by
\begin{equation}
    \rho_{AB} = \sum_{m, n = 0}^1 p_{mn} |\gamma_{mn}\rangle\langle \gamma_{mn}|,
\end{equation} 
where $|\gamma_{mn}\rangle = \frac{1}{\sqrt{2}}(|0, n\rangle + (-1)^m|1, 1\oplus n\rangle)$ and $\sum_{m, n =0}^1 p_{mn} = 1, p_{mn} \ge 0$. An equivalent characterization of Bell diagonal states in terms of Bloch vectors and correlation matrix is given by
\begin{equation}
    \rho_{AB} = \frac{1}{4} [I_4 + \sum_{i, j = 1}^3 t^b_{ij} \sigma_i \otimes \sigma_j],
\end{equation} where $T^b = [t^b_{ij}]$ with $t^b_{ii} = c_i, -1 \le c_i \le 1$ and $t^b_{ij} = 0$ for $i \ne j$. The probabilities $p_{mn}$ are related to the correlation matrix via the following relation:
\begin{equation}
p_{mn} = \frac{1}{4}(1 + (-1)^m c_1 -(-1)^{m+n}c_2 + (-1)^n c_3.
\end{equation} 
The plots in figure \ref{fig:NCEB_BD} depict the conditional entropy of Bell diagonal states before and after the action of $id \otimes \mathcal{N}^2_p$ for $p = \frac{1}{2}$. It is observed that the states possessing negative conditional entropy are concentrated in the corners of the Bell diagonal tetrahedron (see figure \ref{fig:NCEB_BD} a). After the action of the channel $id \otimes \mathcal{N}^2_p$, the conditional entropy of the output states, especially those corresponding to the corners of the tetrahedron, turns non-negative (figure \ref{fig:NCEB_BD} b) This corroborates with the NCEB nature of $id \otimes \mathcal{N}^2_p$ when $p = \frac{1}{2}$\\
\begin{figure}[hbp]
  \centering
  \includegraphics[width=\columnwidth]{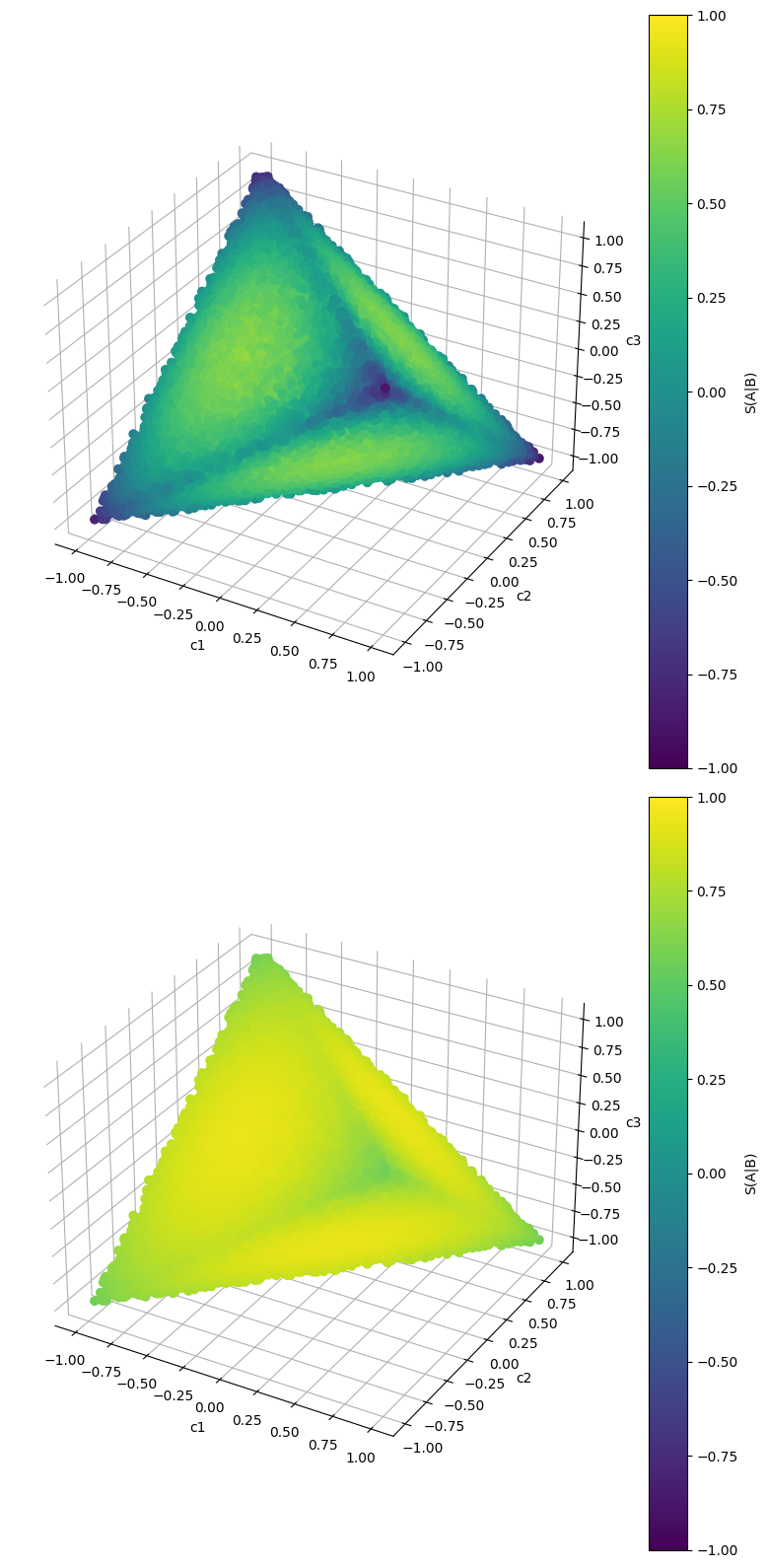}
\caption{Conditional entropy of Bell diagonal states under the action of the channel $id \otimes \mathcal{N}^2_p$ for $p = \frac{1}{2}$. (a) Initial conditional entropy distribution, showing states with negative conditional entropy concentrated in the corners of the Bell diagonal tetrahedron. (b) Conditional entropy after the action of the channel, illustrating the shift to non-negative values, especially for states initially located at the tetrahedron corners.}
 \label{fig:NCEB_BD}
\end{figure}

\noindent \textbf{II. Example of NCEA channel:}\label{sec:NCEBA_Ex_NCEA} This section presents two quantum channels in the context of NCEA channels and determines the conditions under which they qualify as annihilating channels. Annihilation is made possible by applying either a global channel on the subsystem or local channels acting on the subsystem. We exhibit the action through the use of global channels.\\ 

\noindent \textit{a. Global depolarizing channel:} A global depolarizing channel acting on a $d^2$-dimensional system $B$ (or alternately, a $d\otimes d$ system $B_1 B_2$) with density operator $\rho \in D(\mathcal{H}_B)$, is described by
\begin{equation}
    \mathcal{E}_{gd}(\rho) = p \rho + (1-p)\frac{I_{d^2}}{d^2}.
    \label{eqn:global_dep_map}
\end{equation} 
Consider $|\phi^+\rangle$ to be a $d \otimes d$ maximally entangled state where $|\phi^+\rangle = \frac{1}{\sqrt{d}} \sum_{i = 1}^d |i\rangle |i\rangle$. The action of the global depolarizing channel on this state is given by
\begin{equation}
    \chi = \mathcal{E}_{gd}(|\phi^+\rangle\langle\phi^+|) = p |\phi^+\rangle\langle\phi^+|  + (1-p)\frac{I_{d^2}}{d^2}.
    \label{eqn:global_dep_map_choi}
\end{equation} The marginal density operators $\chi_{B_1}, \chi_{B_2}$ of $\chi$ simplify to $\frac{I_d}{d}$. 

We establish a sufficiency condition for the depolarizing channel to be NCEA using the maximally entangled state. The action of the depolarizing channel on $|\phi^+\rangle$ produces an isotropic state, with $d^2$ eigenvalues: $\lambda_1 = \frac{1 + (d^2-1)p}{d^2}$ and $\lambda_2 = \dots = \lambda_{d^2} = \frac{1-p}{d^2}$. Thus, the Von Neumann entropy of $\chi$ is
\begin{equation}
\begin{split}
    S(\chi) & = -\frac{1 + (d^2-1)p}{d^2}\log(\frac{1 + (d^2-1)p}{d^2}) \\ & -(d^2 - 1)\frac{1-p}{d^2}\log(\frac{1-p}{d^2}) = S(p,d).
\end{split}
\end{equation}
Finally, conditional entropy of the output system $B$ across a fixed partition $B_1 - B_2$ is
\begin{equation}
    S(B_1|B_2)_\chi = S(p, d) - \log d.
\end{equation}

\begin{proposition}
    \label{prop:max_ent_gdp}
    Let $\Omega$ represent a maximally entangled state in a $d \otimes d$ system except $|\phi^+\rangle$, and let $\chi'$ be its output state under the action of the global depolarizing channel. Then $S(B_1|B_2)_\chi = S(B_1|B_2)_{\chi'}$.
\end{proposition} 
The equality in conditional entropy follows from the fact that $\chi'$ has the same spectra as $\chi$. Also, both $\chi$ and $\chi'$ have the same marginal density operator, $\frac{I_d}{d}$. 

\begin{proposition}
    \label{prop:pure_state_gdp}
    Let $\Pi$ be a pure state in the Hilbert space of the $B$, and $\chi_\pi$ be the corresponding output state under the channel. It follows that, $S(\chi) = S(\chi_\pi)$
\end{proposition} 
\noindent The equality is due to $\chi$ and $\chi_\pi$ having the same eigenvalues. However, the marginals density operators of $\chi$ and $\chi_\pi$ are not equal with $S(\chi)_{B_2} \ge S(\chi_\pi)_{B_2}$. This gives rise to the inequality
\begin{equation}
    S(B_1|B_2)_\chi \le S(B_1|B_2)_{\chi_\pi}
\end{equation}
\noindent Combining propositions \ref{prop:max_ent_gdp} and \ref{prop:pure_state_gdp}, we arrive at our main statement
\begin{lemma}
\label{lemma:glob_dep_cond}
    If a global depolarizing channel annihilates the conditional entropy of a maximally entangled state, it does so for all pure input states. In other words, if $S(B_1|B_2)_\chi = S(\chi) - S(\chi)_{B_2} \ge 0$, then $S(B_1|B_2)_{\chi_\pi} = S(\chi_\pi) - S(\chi_\pi)_{B_2} \ge 0$
    \label{lemma:global_dep}
\end{lemma}

Thus, if $\mathcal{E}_{gd}$ annihilates the conditional entropy of a maximally entangled state, then it does so for all states, pure or mixed. Consequently, verifying the action of the global depolarizing channel on a pure maximally entangled state suffices to comment on whether it is NCEA or not.

To illustrate the application of Lemma \ref{lemma:global_dep}, consider a $2 \otimes 2$ global depolarizing channel based on equation \eqref{eqn:global_dep_map}. When this channel acts on a Bell state $|\phi^{+}\rangle = \frac{1}{\sqrt{2}}(|00\rangle + |11\rangle)$, the output states are separable for $0 \le p \le \frac{1}{3}$. In contrast, the range of non-negative conditional entropy for the output states is $0 \le p \le 0.748$. By Lemma \ref{lemma:global_dep}, we can conclude that the $2 \otimes 2$ global depolarizing channel is NCEA for $p \le 0.748$, a broader range than the separability condition.\\

\noindent \textit{b. Transpose depolarizing channel:} The action of a transpose depolarizing channel on a $d^2 \times d^2$ complex matrix $\mu$ is given by
\begin{equation}
    \Phi(\mu) = t \mu^T  + (1-t) \frac{I_{d^2}}{d^2} Tr\{\mu\}
\end{equation} where $t \in [\frac{-1}{d-1}, \frac{1}{d+1}]$, and $\mu^T$ represents the full transpose of the complex matrix. For a density matrix $\rho$, $\rho^T$ is a valid state, and hence the transpose depolarizing channel is a valid completely positive, trace-preserving (CPTP) map. It follows that
\begin{equation}
    \Phi(\rho) = t \rho^T + (1 -t) \frac{I_{d^2}}{d^2}. 
\end{equation} $\Pi^T$ is pure for a pure state $\Pi$. Hence, the action of the transpose depolarizing channel reduces to that of the global depolarizing channel discussed above. Therefore, the sufficiency condition in the lemma \ref{lemma:glob_dep_cond} applies to transpose depolarizing channels.\\

\textit{Note:} It should be noted that 2-local qubit depolarizing channel $\mathcal{N}^D \otimes \mathcal{N}^D$ (where $\mathcal{N}^D(\rho) = p \rho + (1-p) \frac{I}{2}$) becomes NCEA  for $p < 0.87$, through its action on $2$-qubit pure entangled states. On the other hand, the same channel $id \otimes \mathcal{N}^D$ is NCEB for a range $p < 0.809$ \cite{divincenzo_1998, wilde_qit2016}. The example is crucial to understanding the difference between NCEA and NCEB.

\section{Relating NCEB and NCEA with other channels}\label{sec:relations}
We establish connections between NCEB, NCEA,  and other quantum channels analytically and through examples. Specifically, the relation between NCEB channels and Coherent Information  Breaking (CIB) channels, Entanglement Breaking (EB) channels, and zero quantum capacity channels are covered. Additionally, we introduce Mutual Information Breaking (MIB) channels, a new class of quantum channels where mutual information vanishes in the output states and demonstrate its connection to NCEB channels. For NCEA channels, we explore their relation with NCVE channels, a set of quantum channels that are non-decreasing in quantum conditional entropy. 

\subsection{Equivalence of NCEB and CIB} \label{sec:relations_CIB}
Here, we cover a characterization of NCEB channels in terms of their coherent information. Zero coherent information channels or Coherent Information Breaking (CIB) channels are those channels $\mathcal{N}_{B\rightarrow \tilde{B}} $ for which the coherent information is zero, i.e., $Q(\mathcal{N}_{B\rightarrow \tilde{B}} ) = 0$. Let the set of such channels mapping states between $d$-dimensional systems be denoted by
\begin{equation}
\mathcal{Q}^{(d)}_0 = \{\mathcal{N}_{B\rightarrow \tilde{B}} | Q(\mathcal{N}_{B\rightarrow \tilde{B}} ) = 0\}.    
\end{equation}
Our key claim is that $NCEB^{(d)} = \mathcal{Q}^{(d)}_0$, the details of which are discussed below.
\begin{lemma}
\label{lemma:pure_state}
 If a channel breaks negative conditional entropy on pure states, it will break negative conditional entropy for all states.
\end{lemma}
\begin{proof}
We know that $\mathcal{D}(\mathcal{H}_{AB})$ is a convex and compact set with pure states being located on the boundary.Let $\Pi_i$ represent a pure state in $D(\mathcal{H}_{AB})$. Thus, any $\rho_{AB} \in D(\mathcal{H}_{AB})$ is expressed as 
\begin{equation}
\rho_{AB} = \sum_i p_i \Pi_i    
\end{equation}
where all $p_i \ge 0$ and $\sum_i p_i =1$. The action of a  given channel $\mathcal{N}_{B\rightarrow \tilde{B}}$ on one part of a bipartite system $\rho_{AB}$ is given by
\begin{equation}
(id_A \otimes \mathcal{N}_{B\rightarrow \tilde{B}})(\rho_{AB}) = \sum_i p_i (id_A \otimes\mathcal{N}_{B\rightarrow \tilde{B}})(\Pi_i).    
\end{equation} Suppose that $\sigma_i = (id_A \otimes \mathcal{N}_{B\rightarrow \tilde{B}})(\Pi_i)$. Consequently, any output state $\sigma_{A\tilde{B}}$ of the channel is given by
\begin{equation}
\begin{split}
    \sigma_{A\tilde{B}} & = (id_A \otimes \mathcal{N}_{B\rightarrow \tilde{B}})(\rho_{AB})\\
    & = \sum_i p_i (id_A \otimes \mathcal{N}_{B\rightarrow \tilde{B}})(\Pi_i)\\
    & = \sum_i p_i \sigma_i
\end{split}
\end{equation} Through the concavity of conditional entropy, we have 
\begin{equation}
    S(A|\tilde{B})_\sigma \ge \sum_i p_i S(A|\tilde{B})_{\sigma_i}
\end{equation}
\noindent Therefore, it suffices to characterize negative conditional entropy breaking channels on pure states.
\end{proof}

\begin{lemma}
\label{lemma:coh}
    The coherent information of a channel $\mathcal{N}_{B\rightarrow \tilde{B}}$ depends on the conditional entropy of the output bipartite state $\sigma_{A\tilde{B}}$, given a pure state input.
\end{lemma}

\begin{proof}
Consider any pure state $\Pi_i \in D(\mathcal{H}_{AB})$ and let $\sigma_{A\tilde{B}} = (id_A \otimes \mathcal{N}_{B\rightarrow \tilde{B}})(\Pi_i)$.  Coherent information of the output state $\sigma_{A\tilde{B}}$ is given by 
\begin{equation}
 I(A\rangle \tilde{B})_\sigma = S(\tilde{B})_\sigma - S(A\tilde{B})_\sigma = -S(A|\tilde{B})_\sigma  
\end{equation}
 and overall, the coherent information of the channel is 
 \begin{equation}
 Q(\mathcal{N}_{B\rightarrow\tilde{B}}) = \max_{\Pi_i} I(A \rangle \tilde{B})_\sigma.    
 \end{equation}
Thus, for pure states, the conditional entropy of the output states is related to coherent information of the channel.
\end{proof}
    
\begin{theorem}
The set of all negative conditional entropy breaking channels is the same as zero coherent information channels, i.e., $NCEB^{(d)} = \mathcal{Q}^{(d)}_0$.
\end{theorem}

\begin{proof}
First, we show that $NCEB^{(d)} \subseteq \mathcal{Q}^{(d)}_0$ and then  $\mathcal{Q}^{(d)}_0\subseteq NCEB^{(d)}$.\\

\noindent Forward implication: $NCEB^{(d)} \subseteq \mathcal{Q}^{(d)}_0$:\\
Based on lemma \ref{lemma:pure_state}, it suffices to consider the action of the channel on pure states to characterize $NCEB^{(d)}$ for all $d \otimes d$ bipartite states. Let a channel $\mathcal{N}_{B\rightarrow \tilde{B}} \in NCEB^{(d)}$. By definition, $S(A|\tilde{B})_\sigma \ge 0, \sigma_{A\tilde{B}} = (id_A \otimes \mathcal{N}_{B\rightarrow \tilde{B}})(\rho_{AB})$ for any input state $\rho_{AB}$. Let $\chi_{A\tilde{B}} = (id_A \otimes \mathcal{N}_{B\rightarrow \tilde{B}})(\Pi_i)$ for some pure state $\Pi_i \in D(\mathcal{H}_{AB})$ and $I(A \rangle \tilde{B})_\chi$ be the coherent information of the output state. From definition, $I(A\rangle \tilde{B})_\chi = -S(A|\tilde{B})_\chi$. It follows that $I(A\rangle \tilde{B})_\chi \leq 0$ as $S(A|\tilde{B})_\chi \geq 0$.

\noindent From lemma \ref{lemma:coh}, it is clear that $Q(\mathcal{N}_{B\rightarrow \tilde{B}})$ must be evaluated using $I(A\rangle \tilde{B})$ for all output states $\chi$. Since the outputs of all pure state inputs have $I(A \rangle \tilde{B})_\chi \le 0$, it follows that $Q(\mathcal{N}_{B\rightarrow \tilde{B}})$ must be $0$ as coherent information of the channel is always non-negative. Hence $\mathcal{N}_{B\rightarrow \tilde{B}} \in \mathcal{Q}^{(d)}_0$. Since this holds true for an arbitrary channel in $NCEB^{(d)}$, it must hold true for all channels in the set. Therefore,
\begin{equation}
NCEB^{(d)} \subseteq \mathcal{Q}^{(d)}_0   . 
\end{equation}
\noindent Reverse implication: $\mathcal{Q}^{(d)}_0 \subseteq NCEB^{(d)}$:\\
Proof by contradiction - Assume that $\mathcal{Q}^{(d)}_0 \not \subseteq NCEB^{(d)}$ and let $\mathcal{N}_{B\rightarrow \tilde{B}} \in \mathcal{Q}^{(d)}_0$ but not in $NCEB^{(d)}$ This implies there exists a state $\rho_{AB} \in \mathcal{D}(\mathcal{H}_{AB})$ such that $S(A|\tilde{B})_\sigma < 0, \sigma_{A\tilde{B}} = (id_A \otimes \mathcal{N}_{B\rightarrow \tilde{B}})(\rho_{AB})$. We know that $\rho_{AB} = \sum_i p_i \Pi_i$ for pure states $\Pi_i$ $\in D(\mathcal{H}_{AB})$ with $\sum_i p_i =1$. Using this and the concavity of conditional entropy, we have that 
\begin{equation}
\begin{split}
    S(A|\tilde{B})_\sigma & = S(A|\tilde{B})_{\sum_i p_i \sigma_i} \\
    & \le S(A|\tilde{B})_{\sigma_i} < 0.
\end{split}
\end{equation}    
\noindent Thus, at least one pure state input exists for which the output state, under the channel, has negative conditional entropy. Recall that $Q(\mathcal{N}_{B\rightarrow \tilde{B}}) = \max_{\Pi_i} I(A \rangle \tilde{B})_{\sigma_i}$ and that $I(A\rangle \tilde{B})_{\sigma_i} = -S(A|\tilde{B})_{\sigma_i}$. It follows that if the output state of the channel for any pure state input has negative conditional entropy, it would imply that the coherent information of the channel is strictly positive and contradicts the fact that $\mathcal{N}_{B\rightarrow \tilde{B}} \in \mathcal{Q}^{(d)}_0$. Therefore, our assumption is incorrect, and we have 
\begin{equation}
NCEB^{(d)} = \mathcal{Q}^{(d)}_0 .
\end{equation}
\end{proof}

\subsection{Relation between NCEB  and MIB} \label{sec:relations_MIB}
We introduce a new type of quantum channel called Mutual Information Breaking (MIB) channel. Such a channel $\mathcal{N}_{B\rightarrow \tilde{B}}$ acts on one part of a $d \otimes d$ system and produces an output state whose quantum mutual information is zero i.e., $I(A; \tilde{B})_\sigma = 0 \text{ where } \sigma_{A\tilde{B}} =  (id_A \otimes \mathcal{N}_{B\rightarrow\tilde{B}})(\rho_{AB})$. Some examples include the total depolarizing channel and the total amplitude damping channel. When applied to a one-part of bipartite state, they transform it into a product state and reduce the mutual information to zero.

The set of MIB channels, denoted by $I^{(d)}_0$, includes all channels $\mathcal{N}_{B\rightarrow \tilde{B}}$ on $d$-dimensional systems that zero out mutual information, defined as
\begin{equation}
\begin{split}
    I^{(d)}_0 = \{\mathcal{N}_{B\rightarrow \tilde{B}} | I(A; \tilde{B})_\sigma = 0 \text{ with }\\
    \sigma_{A\tilde{B}} = (id_A \otimes \mathcal{N}_{B\rightarrow\tilde{B}})(\rho_{AB}),\\
    \forall \rho_{AB} \in D(\mathcal{H}_{AB})\}
\end{split}
\end{equation}
We prove that any mutual information-breaking channel is a conditional entropy breaking channel in the following theorem.
\begin{theorem}
$I^{(d)}_0 \subset NCEB^{(d)}$   
\end{theorem}
\begin{proof}
Let $\mathcal{N}_{B\rightarrow \tilde{B}} \in I^{(d)}_0$ be arbitrary channel. The action of the channel on a $d \otimes d$ state $\rho_{AB}$ results in a state $\sigma_{A\tilde{B}}$. By definition, the mutual information of the resultant state will be,
\begin{eqnarray}
I(A: \tilde{B})_\sigma=S(A)_\sigma +S(\tilde{B})_\sigma -S(A \tilde{B})_\sigma = 0\nonumber\\
\implies S(A\tilde{B})_\sigma =S(A)_\sigma +S(\tilde{B})_\sigma.
\end{eqnarray}
The conditional entropy of the output state $\rho_{A\tilde{B}}$ in such a scenario will be
\begin{equation}
S(A| \tilde{B})_\sigma = S(A\tilde{B})_\sigma - S(\tilde{B})_\sigma = S(A)_\sigma \ge 0 .  
\end{equation}
\noindent This follows from the fact that the entropy of the quantum state is non-negative. Since this holds true for an arbitrary input state $\rho_{AB}$, we conclude that  $\mathcal{N}_{B\rightarrow \tilde{B}} \in NCEB^{(d)}$. Hence $I^{(d)}_0 \subset NCEB^{(d)}$
\end{proof}
However, the converse is not true i.e $NCEB^{(d)} \not \subseteq I^{(d)}_0 $. This is because quantum states with non-negative conditional entropy may also have positive values for mutual information.

\subsection{Relation between NCEB  and EB} \label{sec:relations_EB}
Entanglement Breaking (EB) channels are a crucial class of quantum channels in quantum information theory \cite{horodecki2003}. An EB channel $\mathcal{N}_{B\rightarrow\tilde{B}}$ is defined as:
\begin{align}
(id_A \otimes \mathcal{N}_{B\rightarrow\tilde{B}})(\sigma_{AB}) = \sum_i p_i \rho^i_A \otimes \rho^i_B,
\end{align}
for all $\sigma_{AB} \in \mathcal{D}(\mathcal{H}_{AB})$. The following theorem establishes a connection between EB channels and NCEB channels.

\begin{theorem}
Given an entanglement breaking channel $\mathcal{N}_{B\rightarrow\tilde{B}}$ acting on a $d$-dimensional system, $\mathcal{N}_{B\rightarrow\tilde{B}}$ is also a NCEB channel.
\end{theorem}

\begin{proof}
Since $\mathcal{N}_{B\rightarrow\tilde{B}}$ is an entanglement breaking channel, we express it as:
\begin{align}
(id_A \otimes \mathcal{N}_{B\rightarrow\tilde{B}})(\rho_{AB}) = \sum_i p_i \sigma^i_A \otimes \sigma^i_B ,
\end{align}
for some input state $\rho_{AB}$. Now, for any separable state, the conditional entropy is non-negative. Thus, by the definition of NCEB channels, $\mathcal{N}_{B\rightarrow\tilde{B}}$ is also an NCEB channel.
\end{proof}
\noindent The above result shows that the set of EB channels is a strict subset of NCEB channels. The subset relation arises because the range space of any EB channel is a strict subset of the set of states possessing non-negative conditional entropy \cite{vempati2021}. This relationship is also stated in a related work \cite{muhuri2023}.

\subsection{Relation between NCEB and Zero Capacity Channels} \label{sec:relations_ZC}
Zero capacity channels are an interesting class of channels in quantum communication and quantum information theory. Such channels do not possess the ability to transmit quantum information. While a complete characterization of zero capacity is still open, two well-known classes of quantum channels, namely positive partial transposition (PPT) channels \cite{Horodecki1996Sep,horodecki1997}, and anti-degradable channels (ADG) channels \cite{cubitt2008} are found to have zero capacity.

\begin{figure}[ht]
  \fbox{\includegraphics[width=7.5cm]{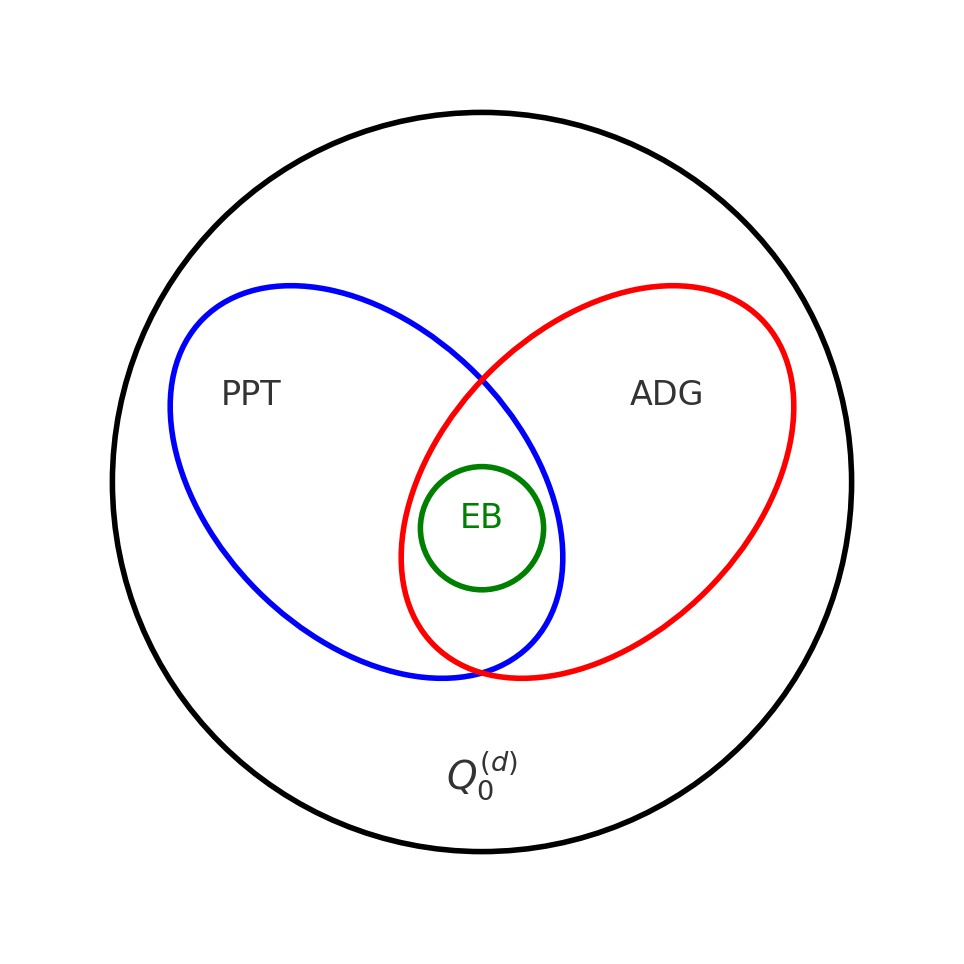}}
  \caption{The above figure depicts the relation with $Q^{(d)}_0$ with well-known zero capacity quantum channels, namely - PPT channels and anti-degradable (ADG) channels. Given the equivalence of $NCEB^{(d)}$ with $Q^{(d)}_0$, it implies that zero capacity quantum channels are also NCEB}
  \label{fig:adg_ppt}
\end{figure}

PPT channels are those quantum channels whose Choi state is a PPT state. Meanwhile, anti-degradable channels allow their output to be simulated from the complementary channel's output. Thus, whatever information was lost to the environment is sufficient to recreate what was sent through the channel, making the actual transmission almost irrelevant for information retrieval. The statement below establishes the connection between NCEB channels and zero capacity quantum channel
\begin{theorem}
  A quantum channel $\mathcal{N}_{B \rightarrow \tilde{B}}$ with zero quantum capacity is also an NCEB channel, i.e., $\mathcal{N}_{B \rightarrow \tilde{B}} \in NCEB^{(d)}$
\end{theorem}
\begin{proof}
Consider a quantum channel $\mathcal{N}_{B \rightarrow \tilde{B}}$ acting on a $d$-dimension system with quantum capacity $\mathscr{Q}(\mathcal{N}_{B \rightarrow \tilde{B}}) = 0$. The coherent information of a quantum channel is non-negative and provides a lower bound for its quantum capacity, i.e., $\mathscr{Q}(\mathcal{N}) \ge  Q(\mathcal{N})$. It follows that
\begin{eqnarray}
    \mathscr{Q}(\mathcal{N}_{B\rightarrow \tilde{B}}) = 0 \implies Q(\mathcal{N}_{B \rightarrow \tilde{B}}) = 0.
\end{eqnarray}
Thus, $\mathcal{N}_{B\rightarrow\tilde{B}} \in \mathcal{Q}^{(d)}_0$ and hence, $\mathcal{N}_{B\rightarrow\tilde{B}} \in NCEB^{(d)}$
\end{proof}

\subsection{Relation between NCEA and NCVE channels} \label{sec:relations_NCVE}
We focus this discussion on the set of channels that are non-decreasing in quantum conditional entropy, referred to as NCVE. Let $AB$ be a quantum system with density operator $\rho_{AB} \in D(\mathcal{H}_{AB})$. Consider a channel $\mathcal{N}_{AB \rightarrow CD}$ and let $\sigma_{CD} = \mathcal{N}_{AB \rightarrow CD}(\rho_{AB})$. Such a channel is said to be NCVE if $S(C|D)_\sigma \ge S(A|B)_\rho$ for all $\rho_{AB} \in D(\mathcal{H}_{AB})$, and the set of all such channels is denoted by $NCVE(A|B \rightarrow C|D)$. $A$-unital represents another class of channels, where a channel $\mathcal{N}_{AB \rightarrow AB}$ is $A$-unital if for every $\rho_{B} \in D(\mathcal{H}_{B}), \mathcal{N}_{AB \rightarrow AB}(\frac{I_A}{d_A} \otimes \rho_B) = \frac{I_A}{d_A} \otimes \sigma_B$. It has been established that $NCVE(A|B)\text{ }(\text{i.e. }NCVE(A|B \rightarrow A|B)) = UNI(A|B)$ \cite{vempati2022}.

Let $\mathcal{N}_{B\rightarrow\tilde{B}}$ be a quantum channel in $NCEA^{(d)}$ where $B, \tilde{B}$ are $d$-dimensional systems. Let $B_1, B_2$, and $\tilde{B}_1, \tilde{B}_2$ represent a fixed bipartition of $B$ and $\tilde{B}$ respectively. Thus, the set of conditional entropy non-decreasing channels is denoted by $NCVE(B_1|B_2)$, and $A$-unital by $UNI(B_1|B_2)$. We touch upon an example of a channel common to these classes and proceed to demonstrate a nuanced relationship between $NCEA^{(d)}$ and $NCVE(B_1|B_2)$.

In section \ref{sec:NCEBA_Ex}, we examined the global depolarizing channel (equation \ref{eqn:global_dep_map}) as an example of NCEA channels. The action of such a map on a $d \otimes d$ system $B$  and obtain the following expression,
\begin{equation}
\begin{split}
    \mathcal{E}_{gd}(\frac{I_{B_1}}{d} \otimes \rho_{B_2}) & = p(\frac{I_{B_1}}{d} \otimes \rho_{B_2}) + (1-p) \frac{I_{B_1}}{d} \otimes \frac{I_{B_2}}{d}\\
    & = \frac{I_{B_1}}{d} \otimes (p \rho_{B_2} + (1-p) \frac{I_{B_2}}{d})\\
    & = \frac{I_{B_1}}{d} \otimes \sigma_{B_2}.
\end{split}
\end{equation}
Thus, the global depolarizing channel is $A$-unital and, consequently, is non-decreasing in conditional entropy. Furthermore, through lemma \ref{lemma:global_dep} and the action of the channel on a maximally entangled state (refer equation \ref{eqn:global_dep_map_choi}), we determine that the range for \( p \) aligns with the conditions for which the output state maintains non-negative conditional entropy.

Naturally, one could ask whether $NCEA^{(d^2)}$ and $NCVE(B_1|B_2)$ are equivalent or possess a intricate relationship. Despite the aforementioned example, we demonstrate that neither set is a subset of the other below:
\begin{itemize}
    \item Membership in $NCVE(B_1|B_2)$ does not necessarily imply inclusion in $NCEA^{(d^2)}$. While $\mathcal{N} \in NCVE(B_1|B_2)$ ensures $S(B_1|B_2)_\sigma \ge S(B_1|B_2)_\rho$, where $\sigma_B = \mathcal{N}(\rho)_B$, it does not guarantee $S(B_1|B_2)_\sigma \ge 0$. For instance, the identity channel $id_B$, included in $NCVE(B_1|B_2)$, does not qualify as an NCEA channel.
    \item Conversely, $NCEA^{(d)}$ is not a subset of $NCVE(B_1|B_2)$ since a non-negative conditional entropy for output states does not imply a greater conditional entropy compared to the input state.
\end{itemize}

\section{Characterization and detection of NCEA and NCEB}\label{sec:convex_compact}
In this section, we discuss the topological characterization of the NCEB and NCEA channels. We show that set $NCEB^{(d)}$ is convex and compact. This confirms the existence of a witness to detect channels that are not in $NCEB^{(d)}$. Identifying such channels is always beneficial in increasing or preserving negative conditional entropy, which is a resource. However, finding such a witness operator is computationally difficult, similar to the entanglement-separability problem. Similarly, we show that the set $NCEA^{(d)}$ is convex.

\subsection{Convexity and compactness of \texorpdfstring{$NCEB^{(d)}$}{NCEB{(d)}}}
\begin{theorem}
    $NCEB^{(d)}$ is \emph{convex}.
\end{theorem}
\begin{proof}
Let $AB$ be a $d \otimes d$ bipartite system. Consider two channels, $\mathcal{N}^1_{B\rightarrow \tilde{B}}, \mathcal{N}^2_{B\rightarrow \tilde{B}}$ from the set $NCEB^{(d)}$. From the definition of $NCEB^{(d)}$, we have
\begin{equation}
\begin{split}
    S(A|\tilde{B})_{\sigma^1} \ge 0 \\
    S(A|\tilde{B})_{\sigma^2} \ge 0  
\end{split}
\end{equation} where $\sigma^1_{A\tilde{B}} = id_A\otimes\mathcal{N}^1_{B\rightarrow \tilde{B}}(\rho_{AB})$ and $\sigma^2_{A\tilde{B}} = id_A\otimes\mathcal{N}^2_{B\rightarrow \tilde{B}}(\rho_{AB})$ and $\rho_{AB} \in D(\mathcal{H}_{AB})$.\\

\noindent Consider $\mathcal{N}_{B\rightarrow \tilde{B}} = \lambda \mathcal{N}^1_{B\rightarrow \tilde{B}} + (1 -\lambda) \mathcal{N}^2_{B\rightarrow \tilde{B}}$ for some $\lambda \in [0,1]$. For any input state $\rho_{AB} \in D(\mathcal{H}_{AB})$ ,
\begin{equation}
\begin{split}
    \sigma_{A\tilde{B}} = & \text{ } id_A\otimes\mathcal{N}_{B\rightarrow \tilde{B}}(\rho_{AB})\\
    = & \text{ } \lambda (id_A\otimes\mathcal{N}^1_{B\rightarrow \tilde{B}})(\rho_{AB}) \text{ }  + \\
    & (1 - \lambda)(id_A\otimes\mathcal{N}^2_{B\rightarrow \tilde{B}})(\rho_{AB})\\
    = & \text{ } \lambda \sigma^1_{A\tilde{B}} + (1-\lambda)\sigma^2_{A\tilde{B}}.
\end{split}
\end{equation}
\noindent Here,$\sigma^1_{A\tilde{B}} = id_A\otimes\mathcal{N}^1_{B\rightarrow \tilde{B}}(\rho_{AB})$ and $\sigma^2_{A\tilde{B}} = id_A\otimes\mathcal{N}^2_{B\rightarrow \tilde{B}}(\rho_{AB})$. The equivalences follow from the linearity of quantum channels and the distributive property of addition over tensor product.

\noindent From the definition of $NCEB^{(d)}$, we have 
\begin{eqnarray}
    S(A|\tilde{B})_{\sigma^1} \ge 0 & S(A|\tilde{B})_{\sigma^2} \ge 0.  
\end{eqnarray}
\noindent It is clear that $\sigma^1_{A\tilde{B}}, \sigma^2_{A\tilde{B}} \in CVENN$, the set of 
non-negative conditional entropy states 
Using convexity of $CVENN$ \cite{vempati2021}, it follows that $\sigma_{A\tilde{B}} \in CVENN$. Since this holds true for an arbitrary input  $\rho_{AB}$, we conclude that $\mathcal{N}_{B\rightarrow \tilde{B}} \in NCEB^{(d)}$, thereby proving convexity of the set. 
\end{proof}

\begin{theorem}
    $NCEB ^{(d)}$ is \emph{compact}
\end{theorem}
\begin{proof} While the above statement has been established independently (refer \cite{muhuri2023}), we provide an alternative proof here.

\noindent Let $\Phi_{NCEB^{(d)}}$ be the set of Choi states of quantum channels in  $NCEB^{(d)}$. We claim that $\Phi_{NCEB^{(d)}}$ has a closed range under conditional entropy function, i.e., $S(A|\tilde{B})_{\Phi_{NCEB^{(d)}}} = [0, \log d]$. We establish the claim by providing examples of NCEB channels that achieve the lower and upper bounds. 

\noindent 
The minimum value of $0$ is achieved for a channel describing local projective measurement on the $B$ subsystem of a maximally entangled state. Whereas conditional entropy of $\log d$ is obtained for a completely depolarizing channel acting on $B$. It must be noted that both these channels are entanglement breaking and, hence, conditional entropy breaking (refer section \ref{sec:relations_EB}). 

\noindent Since the output range forms a closed interval and conditional entropy is a uniformly continuous function, it implies that $\Phi_{NCEB^{(d)}}$ is closed. Additionally, as channel-state isomorphism is continuous \cite{sohail2021} and $NCEB^{(d)}$ is the inverse image of the closed set $\Phi_{NCEB^{(d)}}$, it follows that $NCEB^{(d)}$ is closed.

\noindent Finally, it is known that the completely bounded trace norm of quantum channels is $1$ \cite{watrous2018}. Hence, channels in $NCEB^{(d)}$ are bounded maps. As $NCEB^{(d)}$ is closed and bounded, it must be compact.
\end{proof}

Since the set $NCEB^{(d)}$ is closed and compact, according to Hahn Banach \cite{holmes2012} theorem, there exists a linear functional or hyperplane that detects non-NCEB channels. We discuss one such functional below.

 \noindent \textbf{Existence of Witness Operators:} Witnesses to resourceful states are Hermitian operators whose expectation value aids in detecting such states. The motivation is to detect non-NCEB channels, which are useful for preserving negative conditional entropy. Thus, one procedure for detecting a non-NCEB channel $\mathcal{N}$ involves applying an unknown channel $\mathcal{N}$ to one part of a pure state and then using a suitable witness operator on the output state to determine whether it possesses negative conditional entropy (refer \cite{muhuri2023} for details). The key aspect of this procedure lies in identifying the pure state upon which the channel is applied. On the other hand, using the equivalence of $NCEB^{(d)}$ to $d$-dimensional CIB channels (\ref{sec:relations_CIB}), one can use methods to detect channels having positive coherent information \cite{siddhu_2021, singh2021detecting}. The methods require knowledge of the channel in terms of the output and environment dimensions and the eigenspace of the output system, which may not be known prior for an unknown channel.
 
Here, we use the notion of functional built-on distance-based measures to detect non-NCEB channels. Distance-based measures have been used extensively to identify entangled states in the context of the separability-entanglement problem \cite{PITTENGER200247, verdal_1998, Streltsov_2010, pandya2023minimum}. Since the set of separable states is convex and compact, there will be a closest separable state to an entangled state. In a similar analogy, our witness is based on the concept of the nearest NCEB channel. We would, however, also like to point out that finding such a nearest channel is computationally a difficult task similar to the situation in the separability-entanglement problem. The distance-based functional is described below:

\noindent Consider a scalar-valued functional
 \begin{equation}
     W(\Psi) = \max (0, \inf_{\Phi \in NCEB^{(d)}}\|\Psi - \Phi\|_\lozenge)
 \end{equation} where $\|.\|_\lozenge$ represents the diamond norm on quantum channels \cite{watrous2018}. From the definition, if $\Psi \in NCEB^{(d)}$, then $W(\Psi) = 0$ and if $W(\Psi) > 0$ it implies $\Psi \not \in NCEB^{(d)}$. Due to compactness of $NCEB^{(d)}$, there will be a channel $\Phi' \in NCEB^{(d)}$ that is nearest to $\Psi$. Thus, $W(\Psi)$ constitutes a witness for non-NCEB channels. A similar witness can also be constructed to detect non-NCEA channels based on the characterization in the next subsection.
\subsection{Convexity of \texorpdfstring{$NCEA^{(d)}$}{NCEA{(d)}}}
\begin{theorem}
    $NCEA^{(d)}$ is \emph{convex}.
\end{theorem}
\begin{proof}
    Consider two quantum channels $\mathcal{N}^1_{B \rightarrow \tilde{B}}, \mathcal{N}^2_{B \rightarrow \tilde{B}}$ 
    from $NCEA^{(d)}$. For any given input density operator $\rho_B \in D(\mathcal{H}_B)$,
    \begin{equation}
    \begin{split}
        S(\tilde{B}_1|\tilde{B}_2)_{\sigma^1_{\tilde{B}}} \ge 0 \\
        S(\tilde{B}_1|\tilde{B}_2)_{\sigma^2_{\tilde{B}}} \ge 0,
    \end{split}
    \end{equation}
    with $\sigma^1_{\tilde{B}} = \mathcal{N}^1_{B\rightarrow \tilde{B}}(\rho_{B})$ and $\sigma^2_{\tilde{B}} = \mathcal{N}^2_{B\rightarrow \tilde{B}}(\rho_{B})$.\\
    
    \noindent Let $\mathcal{N}_{B\rightarrow\tilde{B}} = \lambda \mathcal{N}^1_{B \rightarrow \tilde{B}} + (1 -\lambda) \mathcal{N}^2_{B \rightarrow \tilde{B}}$ for some $0 \le \lambda \le1$ and consider the action of $\mathcal{N}_{B\rightarrow\tilde{B}}$ on $\rho_B$. Thus, 
    \begin{equation}
    \begin{split}
        \epsilon_{\tilde{B}} = & \mathcal{N}_{B\rightarrow \tilde{B}}(\rho_{B})\\
        = & \lambda \mathcal{N}^1_{B\rightarrow \tilde{B}}(\rho_{B}) + (1 - \lambda) \mathcal{N}^2_{B\rightarrow \tilde{B}}(\rho_{B})\\
        = & \lambda \epsilon^1_{\tilde{B}} + (1-\lambda)\epsilon^2_{\tilde{B}},
    \end{split}
    \end{equation}
    where $\epsilon^1_{\tilde{B}} = \mathcal{N}^1_{B\rightarrow\tilde{B}}(\rho_B)$ and $\epsilon^2_{\tilde{B}} = \mathcal{N}^2_{B\rightarrow\tilde{B}}(\rho_B)$. The conditional entropy across the $\tilde{B}_1 - \tilde{B}_2$ partition is expressed as
    \begin{equation}
    \begin{split}
         S(\tilde{B}_1|\tilde{B}_2)_{\epsilon} & = S(\tilde{B_1}|\tilde{B}_2)_{\lambda \epsilon^1 + (1-\lambda) \epsilon^2}\\
        & \ge \lambda S(\tilde{B}_1|\tilde{B}_2)_{\epsilon^1} + (1-\lambda)S(\tilde{B}_1|\tilde{B}_2)_{\epsilon^2} \\
        & \ge 0.
    \end{split}
    \end{equation}
    The first inequality follows from the concavity conditional entropy and the second from the definition of $NCEA^{(d)}$ channels. Therefore, the set of $NCEA^{(d)}$ must be convex.
\end{proof}
Finally, $NCEA^{(d)}$ is compact, and the proof can be constructed using techniques analogous to the one used in \cite{muhuri2023}.

\section{Conclusions} \label{sec:conclusion}
 In this article, we have dealt with channels that destroy the negative conditional entropy of a quantum state. Since negative conditional entropy is a resource, these channels can be broadly classified as resource-breaking channels. In particular, we have extended the characterization of negative conditional entropy breaking (NCEB) channels in further depth and detail. We have examined the properties of NCEB channels when combined serially and in parallel and investigated their complementary channels. This has provided us insights into the information-leaking nature of NCEB channels. In addition, we have introduced a class of channels called the negative conditional entropy annihilating (NCEA) channels and examined its properties, including serial combination. In this work, we take depolarizing channels to give examples and further characterize these channels.\\ 
 
Furthermore, we have connected NCEB and NCEA channels with standard channels like zero capacity channels, Entanglement Breaking (EB) channels, conditional von Neumann entropy non-decreasing (NCVE) channels, along with newly introduced channels like Coherent Information Breaking (CIB) channels, and Mutual Information Breaking (MIB) channels. We covered a topological characterization of NCEB channels by showing that the set containing them is convex and compact.  This empowers us to detect channels that will not break the negativity of conditional entropy, ensuring the conservation of quantum resources. A complete characterization of the NCEB and NCEA channels in the form of Kraus representation is an important direction for future research.\\ 
 
\noindent \textit{Acknowledgements :} NG acknowledges support from the project grant received from DST-SERB (India) under the MATRICS scheme, vide file number MTR/2022/000101. PV thanks Ms. Mahathi Vempati for useful discussions while working through this problem. PV also thanks Shirley Chauhan and Rutvij Menavlikar for their suggestion for improving the presentation of this work. 
\bibliography{references}
\bibliographystyle{apsrev4-2}

\end{document}